\documentclass[a4paper,oneside]{scrartcl}

\newcommand{\Title}{Complexity Results for Modal Dependence Logic}
\newcommand{\Author}{
Peter Lohmann%
\thanks{
Leibniz University Hannover,
Theoretical Computer Science,
Appelstr.~4, 30167~Hannover, Germany,
\{lohmann,vollmer\}@thi.uni-hannover.de},
Heribert Vollmer\footnotemark[1]
}
\newcommand{\PDFAuthor}{Peter Lohmann, Heribert Vollmer}
\newcommand{\Keywords}{dependence logic, modal logic, satisfiability problem, computational complexity, poor man's logic}
\newcommand{\SubjectClassifiers}{F.2.2 Complexity of proof procedures; F.4.1 Modal logic}

\usepackage{amsmath, amssymb, amsthm}
\usepackage{xspace}
\usepackage{xcolor}
\usepackage{enumerate}
\usepackage{ifthen}
\usepackage{fancyhdr}
\usepackage{listings}
  \lstdefinelanguage{pseudo}{
    morekeywords={if,elseif,then,return,end,choose,guess},
    morecomment=[l]{//}}
\usepackage{tikz}
  \usetikzlibrary{automata}

\theoremstyle{definition}
\newtheorem{definition}{Definition}[section]
\theoremstyle{plain}
\newtheorem{theorem}[definition]{Theorem}
\newtheorem{corollary}[definition]{Corollary}
\newtheorem{lemma}[definition]{Lemma}

\DeclareMathSymbol{\Pi}{\mathalpha}{operators}{"05}
\DeclareMathSymbol{\Sigma}{\mathalpha}{operators}{"06}

\newcommand{\numberClassFont}[1]{\mathbb{#1}}     
\newcommand{\complexityClassFont}[1]{\mathrm{#1}} 
\newcommand{\logicClFont}[1]{\mathsf{#1}}        
\newcommand{\problemFont}[1]{\mathrm{#1}}         
\newcommand{\mathCommandFont}[1]{\mathrm{#1}}     

\newcommand{\N}{\protect\ensuremath{\numberClassFont{N}}\xspace}

\newcommand{\bigO}[1]{\protect\ensuremath{{O(#1)}}}

\newcommand{\co}{{\protect\ensuremath{\complexityClassFont{co}}}}
\newcommand{\trivial}{\protect\ensuremath{\complexityClassFont{trivial}}\xspace}
\newcommand{\coNP}{\protect\ensuremath{\co\NP}\xspace}
\newcommand{\PTIME}{\protect\ensuremath{\complexityClassFont{P}}\xspace}
\newcommand{\NP}{\protect\ensuremath{\complexityClassFont{NP}}\xspace}

\newcommand{\PiTwo}{\protect\ensuremath{\complexityClassFont{\Pi_2^p}}\xspace}
\newcommand{\SigmaTwo}{\protect\ensuremath{\complexityClassFont{\Sigma_2^p}}\xspace}
\newcommand{\SigmaThree}{\protect\ensuremath{\complexityClassFont{\Sigma_3^p}}\xspace}
\newcommand{\PSPACE}{\protect\ensuremath{\complexityClassFont{PSPACE}}\xspace}
\newcommand{\NEXPTIME}{\ensuremath{\complexityClassFont{NEXP}}\-\ensuremath{\complexityClassFont{TIME}}\xspace}

\newcommand{\leqpm}{\protect\ensuremath{\leq^\mathCommandFont{p}_\mathCommandFont{m}}}

\newcommand{\true}{\protect\textnormal{\raisebox{-1pt}{$\top$}}}
\newcommand{\false}{\protect\textnormal{\raisebox{-1pt}{$\bot$}}}
\newcommand{\aneg}{\protect\ensuremath{\overline{\,\cdot\,}}\xspace}
\newcommand{\dep}[2][]{\protect\ensuremath{\mathCommandFont{=}\ifthenelse{\equal{#1}{}}{(#2)}{(#1,#2)}}\xspace}
\newcommand{\norop}{\mathrel{\bigcirc\kern -0.845em\text{\raisebox{-0.2ex}[0ex][0ex]{$\vee$}}}}
\newcommand{\nor}{\protect\ensuremath{\mathord{\norop}}\xspace}
\newcommand{\bignordisplay}[2][]{\protect\ensuremath{\mathord{\underset{#2}{\overset{#1}{\text{{\Large$\norop$}}}}}}}
\newcommand{\bignor}[2][]{\protect\ensuremath{\mathord{\text{{\Large$\norop$}}_{#2}^{#1}}}}

\newcommand{\MDL}[1][]{\ensuremath{\logicClFont{MDL}\ifx#1\else(#1)\fi}\xspace}
\newcommand{\MDLpara}[2][]{\ensuremath{\logicClFont{MDL}_{#2}\ifx#1\else(#1)\fi}\xspace}
\newcommand{\MDLk}[1][]{\MDLpara[#1]{k}}
\newcommand{\CTL}[1][]{\ensuremath{\logicClFont{CTL}\ifx#1\else(#1)\fi}\xspace}
\newcommand{\LTL}[1][]{\ensuremath{\logicClFont{LTL}\ifx#1\else(#1)\fi}\xspace}
\newcommand{\CTLs}[1][]{\ensuremath{\logicClFont{\CTL^*}\ifx#1\else(#1)\fi}\xspace}

\newcommand{\MDLSAT}[1][]{\ensuremath{\problemFont{MDL}\text{-}}\allowbreak\ensuremath{\problemFont{SAT}\ifx#1\else(#1)\fi}\xspace}
\newcommand{\MDLSATpara}[2][]{\ensuremath{\problemFont{MDL}_{#2}\text{-}}\allowbreak\ensuremath{\problemFont{SAT}\ifx#1\else(#1)\fi}\xspace}
\newcommand{\MDLSATk}[1][]{\MDLSATpara[#1]{k}}
\newcommand{\oneinthree}{\protect\ensuremath{\problemFont{R_{1/3}}}\xspace}
\newcommand{\SpecialQBF}{\protect\ensuremath{\problemFont{QCSP_2(}}\allowbreak\ensuremath{\oneinthree)}\xspace}
\newcommand{\cnf}{\protect\ensuremath{\problemFont{3CNF}}\xspace}
\newcommand{\onecnf}{\protect\ensuremath{\problemFont{1CNF}}\xspace}
\newcommand{\qbfcnf}{\protect\ensuremath{\problemFont{\cnf\text{-}QBF_3}}\xspace}
\newcommand{\dqbf}{\protect\ensuremath{\problemFont{DQBF}}\xspace}
\newcommand{\DQBFCNF}{\protect\ensuremath{\problemFont{\cnf\text{-}\dqbf}}\xspace}

\newcommand{\existOperator}{\protect\ensuremath{\exists\cdot}}
\newcommand{\powerset}[1]{\protect\ensuremath{\mathcal{P}}(#1)\xspace}

\newcommand{\atom}{\protect\ensuremath{\texttt{Atomic}}\xspace}
\newcommand{\calC}{\protect\ensuremath{\mathcal{C}}\xspace}

\usepackage[
  colorlinks=false,pdfdisplaydoctitle=false,pdfkeywords={\Keywords},pdftitle={\Title},pdfauthor={\PDFAuthor}
]{hyperref}

\begin{document}

\title{\Title}

\author{\Author}

\maketitle

\thispagestyle{fancy}
\renewcommand{\headrulewidth}{0pt}
\renewcommand{\footrulewidth}{0pt}
\fancyhf{}
\fancyfoot[L]{\scriptsize This work was partly supported by the NTH Focused Research School for IT Ecosystems, by DFG VO 630/6-1, and by a DAAD PPP grant.}

\begin{abstract}
\small
Modal dependence logic was introduced recently by V\"a\"an\"anen. It enhances the basic modal language by an operator $\dep{}$. For propositional variables $p_{1},\dots,p_{n}$, $\dep[p_{1},\dots,p_{n-1}]{p_{n}}$ intuitively states that the value of $p_{n}$ is determined by those of $p_{1},\dots,p_{n-1}$. Sevenster (J.~Logic and Computation, 2009) showed that satisfiability for modal dependence logic is complete for nondeterministic exponential time. 

In this paper we consider fragments of modal dependence logic obtained by restricting the set of allowed propositional connectives. We show that satisfibility for \emph{poor man's dependence logic}, the language consisting of formulas built from literals and dependence atoms using $\wedge$, $\Box$, $\Diamond$ (i.\,e., disallowing disjunction), remains \NEXPTIME-complete. If we only allow monotone formulas (without negation, but with disjunction), the complexity drops to \PSPACE-completeness.

We also extend V\"a\"an\"anen's language by allowing classical disjunction besides dependence disjunction and show that the satisfiability problem remains \NEXPTIME-complete. If we then disallow both negation and dependence disjunction, satistiability is complete for the second level of the polynomial hierarchy.
Additionally we consider the restriction of modal dependence logic where the length of each single dependence atom is bounded by a number that is fixed for the whole logic. We show that the satisfiability problem for this bounded arity dependence logic is \PSPACE-complete and that the complexity drops to the third level of the polynomial hierarchy if we then disallow disjunction.

In this way we completely classifiy the computational complexity of the satisfiability problem for all restrictions of propositional and dependence operators considered by V\"a\"an\"anen and Sevenster.

A short version of this was presented at CSL 2010 \cite{lovo10}.
\end{abstract}

\smallskip
\noindent{\bfseries ACM Subject Classifiers:} \SubjectClassifiers


\section{Introduction}

The concept of extending first-order logic with partially ordered quantifiers, and hence expressing some form of independence between variables, was first introduced by Henkin \cite{he61}. Later, Hintikka and Sandu developed independence friendly logic \cite{hisa89} which can be viewed as a generalization of Henkin's logic. Recently, Jouko V\"a\"an\"anen introduced the dual notion of functional dependence into the language of first-order logic \cite{va07}. In the case of first-order logic, the independence and the dependence variants are expressively equivalent.

Dependence among values of variables occurs everywhere in computer science (data\-bases, software engineering, knowledge representation, AI) but also the social sciences (human history, stock markets, etc.), and thus dependence logic is nowadays a much discussed formalism in the area called \emph{logic for interaction.} 
Functional dependence of the value of a variable $p_{n}$ from the values of the variables $p_{1},\dots,p_{n-1}$ states that there is a function, say $f$, such that $p_{n}=f(p_{1},\dots,p_{n-1})$, i.\,e., the value of $p_{n}$ only depends on those of $p_{1},\dots,p_{n-1}$. We will denote this in this paper by $\dep[p_{1},\dots,p_{n-1}]{p_{n}}$. 

Of course, dependence does not manifest itself in a single world, play, event or observation. Important for such a dependence to make sense is a collection of such worlds,  plays, events or observations. These collections are called \emph{teams}.
They are the basic objects in the definition of semantics of dependence logic.
A team can be a set of plays in a game. Then $\dep[p_{1},\dots,p_{n-1}]{p_{n}}$ intuitively states that in each play, move $p_{n}$ is determined by moves $p_{1},\dots,p_{n-1}$.
A team can be a database. Then $\dep[p_{1},\dots,p_{n-1}]{p_{n}}$ intuitively states that in each line, the value of attribute $p_{n}$ is determined by the values of attributes $p_{1},\dots,p_{n-1}$, i.\,e., that $p_{n}$ is functionally dependent on $p_{1},\dots,p_{n-1}$.
In first-order logic, a team formally is a set of assignments; and $\dep[p_{1},\dots,p_{n-1}]{p_{n}}$ states that in each assignment, the value of $p_{n}$ is determined by the values of $p_{1},\dots,p_{n-1}$.
Most important for this paper, in modal logic, a team is a set of worlds in a Kripke structure; and $\dep[p_{1},\dots,p_{n-1}]{p_{n}}$ states that in each of these worlds, the value of the propositional variable $p_{n}$ is determined by the values of $p_{1},\dots,p_{n-1}$.

Dependence logic is defined by simply adding these dependence atoms to usual first-order logic \cite{va07}. Modal dependence logic (\MDL) is defined by introducing these dependence atoms to modal logic \cite{va08,se09}. The semantics of \MDL is defined with respect to sets $T$ of worlds in a frame (Kripke structure) $W$, for example
$W,T\models \dep[p_{1},\dots,p_{n-1}]{p_{n}}$ if for all worlds $s,t\in T$, if $p_{1},\dots,p_{n-1}$ have the same values in both $s$ and $t$, then $p_{n}$ has the same value in $s$ and $t$, and a formula 
$$\Box\dep[p_{1},\dots,p_{n-1}]{p_{n}}$$
is satisfied in a world $w$ in a Kripke structure $W$, if in the team $T$ consisting of all successor worlds of $w$, $W,T\models \dep[p_{1},\dots,p_{n-1}]{p_{n}}$.

\MDL was introduced in \cite{va08}. V\"a\"an\"anen introduced besides the usual inductive semantics an equivalent game-theoretic semantics. Sevenster \cite{se09} considered the expressibility of \MDL and proved, that on singleton teams $T$, there is a translation from \MDL to usual modal logic, while on arbitrary sets of teams there is no such translation. Sevenster also initiated a complexity-theoretic study of modal dependence logic by proving that the satisfiability problem for \MDL is complete for the class \NEXPTIME of all problems decidable nondeterministically in exponential time.

In this paper, we continue the work of Sevenster by presenting a more thorough study on  complexity questions related to modal dependence logic. A line of research going back to Lewis \cite{le79} and recently taken up in a number of papers \cite{rewa00,he01,hescsc10,memuthvo08} has considered fragments of different propositional logics by restricting the propositional and temporal operators allowed in the language.
The rationale behind this approach is that by systematically restricting the language, one might find a fragment with efficient algorithms but still high enough expressibility in order to be interesting for applications. This in turn might lead to better tools for model checking, verification, etc. On the other hand, it is worthwhile to identify the sources of hardness: What exactly makes satisfiability, model checking, or other problems so hard for certain languages?

We follow the same approach here. We consider all subsets of modal operators $\Box,\Diamond$ and propositional operators $\wedge$, $\vee$, $\aneg$ (atomic negation), $\true, \false$ (the Boolean constants true and false), i.\,e., we study exactly those operators considered by V\"a\"an\"anen \cite{va08}, and examine the satisfiability problem for \MDL restricted to the fragment given by these operators.
Additionally we consider a restricted version of the $\dep{}$ operator in which the arity of the operator is no longer arbitrarily large but bounded by a constant that is fixed for the considered logic.
In each case we exactly determine the computational complexity in terms of completeness for a complexity class such as \NEXPTIME, \PSPACE, \coNP, etc., or by showing that the satisfiability problem admits an efficient (polynomial-time) solution. We also extend the logical language of \cite{va08} by adding classical disjunction (denoted here by $\nor$) besides the dependence disjunction. Connective $\nor$ was already considered by Sevenster (he denoted it by $\bullet$),
but not from a complexity point of view. In this way, we obtain a complexity analysis of the satisfiability problem for \MDL for all subsets of operators studied by V\"a\"an\"anen and Sevenster as well as the arity bounded dependence operator.

Our results are summarized in Table~\ref{results} for dependence atoms of unbounded arity and in Table~\ref{results-bounded} for dependence atoms whose arity is bounded by a fixed $k\geq 3$. Here $+$ denotes presence and $-$ denotes absence of an operator, and $*$ states that the complexity does not depend on the operator. One of our main and technically most involved contributions addresses a fragment that has been called \emph{Poor Man's Logic} in the literature on modal logic \cite{he01}, i.\,e., the language without disjunction $\vee$. We show that for unbounded arity dependence logic we still have full complexity (Theorem~\ref{poor man dep complexity}, first line of Table~\ref{results}),  i.\,e., we show that Poor Man's Dependence Logic is \NEXPTIME-complete. If we also forbid negation, then the complexity drops down to $\SigmaTwo(=\NP^\NP)$; i.\,e., Monotone Poor Man's Dependence Logic is $\SigmaTwo$-complete (Theorem~\ref{poor man bullet complexity}, but note that we need $\nor$ here). And if we instead restrict the logic to only contain dependence atoms of arity less or equal $k$ for a fixed $k\geq 3$ the complexity drops to $\SigmaThree(=\NP^{\SigmaTwo})$; i.\,e., bounded arity Poor Man's Dependence Logic is $\SigmaThree$-complete (Corollary~\ref{bounded dep concrete}b).

\begin{table}[ht]
\begin{center}
\[
\begin{array}{c|c||c|c|c|c|c||c|c||c|p{8.5em}}
\Box&\Diamond&\wedge&\vee&\aneg&\true&\false&\dep{}&\nor&\textbf{Complexity}&\textbf{Reference}\\
\hline
+&+ & +&*&+&*&* & +&* & \NEXPTIME&Theorem~\ref{poor man dep complexity}\\
+&+ & +&+&+&*&* & -&* & \PSPACE&Corollary~\ref{simple cases}a\\
+&+ & +&+&-&*&+ & *&* & \PSPACE&Corollary~\ref{simple cases}b\\
+&+ & +&-&+&*&* & -&+ & \SigmaTwo&Theorem~\ref{poor man bullet complexity}\\
+&+ & +&-&-&*&+ & *&+ & \SigmaTwo&Theorem~\ref{poor man bullet complexity}\\
+&+ & +&-&+&*&* & -&- & \coNP&\cite{la77}, \cite{dolenahonuma92}\\
+&+ & +&-&-&*&+ & *&- & \coNP&Corollary~\ref{simple cases}c\\
\hline
+&- & +&+&+&*&* & *&* & \NP&Corollary~\ref{one modality cases}a\\
-&+ & +&+&+&*&* & *&* & \NP&Corollary~\ref{one modality cases}a\\
+&- & +&-&+&*&* & *&+ & \NP&Corollary~\ref{one modality cases}a\\
-&+ & +&-&+&*&* & *&+ & \NP&Corollary~\ref{one modality cases}a\\
+&- & +&-&+&*&* & *&- & \PTIME&Corollary~\ref{one modality cases}b\\
-&+ & +&-&+&*&* & *&- & \PTIME&Corollary~\ref{one modality cases}b\\
+&- & +&*&-&*&* & *&* & \PTIME&Corollary~\ref{one modality cases}c\\
-&+ & +&*&-&*&* & *&* & \PTIME&Corollary~\ref{one modality cases}c\\
*&* & -&*&*&*&* & *&* & \PTIME&Corollary~\ref{one modality cases}d\\
*&* & *&*&-&*&- & *&* & \trivial&Corollary~\ref{simple cases}d\\
\hline
-&- & +&+&+&*&* & *&* & \NP&\cite{co71}\\
-&- & +&*&+&*&* & *&+ & \NP&\cite{co71}, $\nor \equiv \vee$\\
-&- & *&-&*&*&* & *&- & \PTIME&Corollary~\ref{simple cases}e\\
-&- & *&*&-&*&* & *&* & \PTIME&Corollary~\ref{simple cases}f
\end{array}\vspace*{-1ex}\]
{\scriptsize $+:$ operator present \quad $-:$ operator absent \quad $*:$ complexity independent of operator}
\end{center}
\vspace*{-1ex}
\caption{Complete classification of complexity for fragments of \MDLSAT \newline{\small All results are completeness results except for the \PTIME cases which are upper bounds.}}
\label{results}
\end{table}

\begin{table}[ht]
\begin{center}
\[
\begin{array}{c|c||c|c|c|c|c||c|c||c|p{8.5em}}
\Box&\Diamond&\wedge&\vee&\aneg&\true&\false&\dep{}&\nor&\textbf{Complexity}&\textbf{Reference}\\
\hline
+&+ & +&+&+&*&* & *&* & \PSPACE&Corollary~\ref{bounded dep concrete}a\\
+&+ & +&+&-&*&+ & *&* & \PSPACE&Corollary~\ref{simple cases}b\\
+&+ & +&-&+&*&* & +&* & \SigmaThree&Corollary~\ref{bounded dep concrete}b\\
+&+ & +&-&+&*&* & -&+ & \SigmaTwo&Theorem~\ref{poor man bullet complexity}\\
+&+ & +&-&-&*&+ & *&+ & \SigmaTwo&Theorem~\ref{poor man bullet complexity}\\
+&+ & +&-&+&*&* & -&- & \coNP&\cite{la77}, \cite{dolenahonuma92}\\
+&+ & +&-&-&*&+ & *&- & \coNP&Corollary~\ref{simple cases}c\\
\hline
+&- & +&+&+&*&* & *&* & \NP&Corollary~\ref{one modality cases}a\\
-&+ & +&+&+&*&* & *&* & \NP&Corollary~\ref{one modality cases}a\\
+&- & +&-&+&*&* & *&+ & \NP&Corollary~\ref{one modality cases}a\\
-&+ & +&-&+&*&* & *&+ & \NP&Corollary~\ref{one modality cases}a\\
+&- & +&-&+&*&* & *&- & \PTIME&Corollary~\ref{one modality cases}b\\
-&+ & +&-&+&*&* & *&- & \PTIME&Corollary~\ref{one modality cases}b\\
+&- & +&*&-&*&* & *&* & \PTIME&Corollary~\ref{one modality cases}c\\
-&+ & +&*&-&*&* & *&* & \PTIME&Corollary~\ref{one modality cases}c\\
*&* & -&*&*&*&* & *&* & \PTIME&Corollary~\ref{one modality cases}d\\
*&* & *&*&-&*&- & *&* & \trivial&Corollary~\ref{simple cases}d\\
\hline
-&- & +&+&+&*&* & *&* & \NP&\cite{co71}\\
-&- & +&*&+&*&* & *&+ & \NP&\cite{co71}, $\nor \equiv \vee$\\
-&- & *&-&*&*&* & *&- & \PTIME&Corollary~\ref{simple cases}e\\
-&- & *&*&-&*&* & *&* & \PTIME&Corollary~\ref{simple cases}f
\end{array}\vspace*{-1ex}\]
{\scriptsize $+:$ operator present \quad $-:$ operator absent \quad $*:$ complexity independent of operator}
\end{center}
\vspace*{-1ex}
\caption{Complete classification of complexity for fragments of \MDLSATk for $k\geq 3$ \newline{\small All results are completeness results except for the \PTIME cases  which are upper bounds.}}
\label{results-bounded}
\end{table}

\section{Modal dependence logic}

We will only briefly introduce the syntax and semantics of modal dependence logic here. For a more profound overview consult V\"a\"an\"anen's introduction \cite{va08} or Sevenster's analysis \cite{se09} which includes a self-contained introduction to \MDL.

\subsection{Syntax}
The formulas of \emph{modal dependence logic} (\MDL) are built from a set $AP$ of \emph{atomic propositions} and the \emph{\MDL operators} $\Box$, $\Diamond$, $\wedge$, $\vee$, $\aneg$ (also denoted $\neg$), $\true$, $\false$, $\dep{}$ and~$\nor$.

The set of \emph{\MDL formulas} is defined by the following grammar
\[\begin{array}{l@{\quad}l}
\varphi\; ::=& \true \;\mid\; \false \;\mid\; p \;\mid\; \neg p \;\mid\; \dep[p_1,\dots,p_{n-1}]{p_n} \;\mid\; \neg \dep[p_1,\dots,p_{n-1}]{p_n} \;\mid\\
& \varphi\wedge \varphi \;\mid\; \varphi\vee\varphi \;\mid\; \varphi\nor\varphi \;\mid\; \Box\varphi \;\mid\; \Diamond\varphi,
\end{array}\]
where $n\geq 1$.

All formulas in the first row will sometimes be denoted as \emph{atomic} formulas and formulas of the form $\dep[p_1,\dots,p_{n-1}]{p_n}$ as \emph{dependence atoms}. The \emph{arity} of a dependence atom $\dep[p_1,\dots,p_{n-1}]{p_n}$ is defined as $n-1$ and with \MDLk we denote the set of all \MDL formulas which do not contain dependence atoms of arity greater than $k$.
We sometimes write $\nabla^k$ for $\underbrace{\nabla\dots\nabla}_{k \text{ times}}$ (with $\nabla\in\{\Box,\Diamond\}$, $k\in\N$).

\subsection{Semantics}
A \emph{frame} (or \emph{Kripke structure}) is a tuple $W=(S,R,\pi)$ where $S$ is a non-empty set of \emph{worlds}, $R\subseteq S\times S$ is the \emph{accessibility relation} and $\pi:S \to \powerset{AP}$ is the \emph{labeling function}.

In contrast to usual modal logic, truth of a \MDL formula is not defined with respect to a single world of a frame but with respect to a set of worlds, as already pointed out in the introduction.
The \emph{truth} of a \MDL formula $\varphi$ in an \emph{evaluation set} $T$ of worlds of a frame $W=(S,R,\pi)$ is denoted by $W,T\models \varphi$ and is defined as follows:
\[\begin{array}{l@{\,}l@{\ }c@{\ }lcp{17em}}
i)&W,T&\models&\top&\quad\quad&always holds\\
ii)&W,T&\models&\bot&\quad\text{iff}\quad&$T=\emptyset$\\
iii)&W,T&\models&p&\quad\text{iff}\quad&$p\in\pi(s)$ for all $s\in T$\\
iv)&W,T&\models&\neg p&\quad\text{iff}\quad&$p\notin\pi(s)$ for all $s\in T$\\
v)&W,T&\models&\dep[p_1,\dots,p_{n-1}]{p_n}&\quad\text{iff}\quad&for all $s_1,s_2\in T$ with\\
&&&\multicolumn{3}{r}{\pi(s_1)\cap\{p_1,\dots,p_{n-1}\}=\pi(s_2)\cap\{p_1,\dots,p_{n-1}\}:}\\
&&&&&$p_n\in \pi(s_1)$\quad{iff}\quad$p_n\in \pi(s_2)$\\
vi)&W,T&\models&\neg\dep[p_1,\dots,p_{n-1}]{p_n}&\quad\text{iff}\quad&$T=\emptyset$\\
vii)&W,T&\models&\varphi\wedge\psi&\quad\text{iff}\quad&$W,T\models \varphi$ and $W,T\models \psi$\\
viii)&W,T&\models&\varphi\vee\psi&\quad\text{iff}\quad&there are sets $T_1,T_2$ with $T=T_1\cup T_2$,\\
&&&&&$W,T_1\models\varphi$ and $W,T_2\models \psi$\\
ix)&W,T&\models&\varphi\nor\psi&\quad\text{iff}\quad&$W,T\models \varphi$ or $W,T\models \psi$\\
x)&W,T&\models&\Box \varphi&\quad\text{iff}\quad&$W,\{s'\mid \exists s\in T$ with $(s,s')\in R\}\models \varphi$\\
xi)&W,T&\models&\Diamond \varphi&\quad\text{iff}\quad&there is a set $T'\subseteq S$ such that $W,T'\models \varphi$ and for all $s\in T$ there is a $s'\in T'$ with $(s,s')\in R$
\end{array}\]

Note the seemingly rather strange definition of $vi)$. The rationale for this, given by V\"a\"an\"anen \cite[p.~24]{va07}, is the fact that if we negate $v)$ and maintain the same duality as between $iii)$ and $iv)$ we get the condition
\[\begin{array}{l@{\ }l}
\text{$\forall s_1,s_2\in T$:}&\pi(s_1)\cap\{p_1,\dots,p_{n-1}\}=\pi(s_2)\cap\{p_1,\dots,p_{n-1}\}\\
&\text{and\quad}p_n\in \pi(s_1)\text{ iff } p_n\notin \pi(s_2),
\end{array}\]
and this is only true if $T=\emptyset$.

By $\vee$ we denote dependence disjunction instead of classical disjunction because the semantics of dependence disjunction is an extension of the semantics of usual modal logic disjunction and thus we preserve downward compatibility of our notation in this way. However, we still call the $\nor$ operator ``classical'' because in a higher level context -- where our sets of states are viewed as single objects themselves -- it is indeed the usual disjunction, cf.~\cite{abva08}.

For each $M\subseteq\{\Box,\Diamond,\wedge,\vee,\aneg,\true,\false,\dep{},\nor\}$ define the set of $\MDL[M ]$ ($\MDLk[M ]$) formulas to be the set of \MDL (resp.~\MDLk) formulas which are built from atomic propositions using only operators and constants from $M$.

We are interested in the parameterized decision problems \textbf{MDL-SAT$\mathbf{(M)}$} and \textbf{MDL$_k$-SAT$\mathbf{(M)}$}:
\begin{description}
  \item[Given]A $\MDL[M ]$ (resp.~ $\MDLk[M ]$) formula $\varphi$.
  \item[Question]Is there a frame $W$ and a non-empty set $T$ of worlds in $W$ such that $W,T\models \varphi$?
\end{description}

Note that, as V\"a\"an\"anen already pointed out \cite[Lemma~4.2.1]{va08}, the semantics of \MDL satisfies the \emph{downward closure property}, i.e., if $W,T\models \varphi$, then $W,T'\models \varphi$ for all $T'\subseteq T$. Hence, to check satisfiability of a formula $\varphi$ it is enough to check whether there is a frame $W$ and a single world $w$ in $W$ such that $W,\{w\}\models \varphi$.

As argued in \cite[Proposition~3.10]{va07}, the downward closure property suits the intuition that a true formula expressing dependence should not becoming false when making the team smaller, since if dependence is true in a large set than it is even more so in a smaller set.

\section{Complexity results}

To state the first lemma we need the following complexity operator.
If $\calC$ is an arbitrary complexity class then $\existOperator\calC$ denotes the class of all sets $A$ for which there is a set $B\in\calC$ and a polynomial $p$ such that for all $x$,
\[x\in A \text{ iff there is a }y\text{ with }|y|\leq p(|x|)\text{ and }\langle x,y\rangle\in B.\]
Note that for every class $\calC$, $\existOperator\calC\subseteq \NP^\calC$. However, the converse does not hold in general.
We will only need the following facts: $\existOperator \coNP=\SigmaTwo$, $\existOperator \PiTwo = \SigmaThree$, $\existOperator\PSPACE=\PSPACE$ and $\existOperator \NEXPTIME = \NEXPTIME$.

Our first lemma concerns sets of operators including classical disjunction.

\begin{lemma}\label{bullet distributivity}
Let $M$ be a set of \MDL operators. Then it holds:
\begin{enumerate}[a)]
 \item Every $\MDL[M\cup\{\nor\}]$ ($\MDLk[M\cup\{\nor\}]$) formula $\varphi$ is equivalent to a formula $\bignor[2^{|\varphi|}]{i=1}\,\psi_i$ with $\psi_i \in \MDL[M ]$ (resp.~$\MDLk[M ]$) for all $i\in\{1,\dots,2^{|\varphi|}\}$.
 \item If $\calC$ is an arbitrary complexity class with $\PTIME \subseteq \calC$ and $\MDLSAT[M ]\in\calC$ ($\MDLSATk[M ]\in\calC$) then $\MDLSAT[M\cup\{\nor\}] \in \existOperator\calC$ (resp.~$\MDLSATk[M\cup\{\nor\}] \in \existOperator\calC$).
\end{enumerate}
\end{lemma}
\begin{proof}
a) follows from the distributivity of $\nor$ with all other operators. More specifically $\varphi \star (\psi\nor\sigma)\equiv(\varphi \star \psi)\nor(\varphi \star\sigma)$ for $\star \in\{\wedge,\vee\}$ and $\nabla (\varphi \nor \psi)\equiv (\nabla\varphi)\nor(\nabla\varphi)$ for $\nabla \in\{\Diamond,\Box\}$.\footnote{Interestingly, but not of relevance for our work, $\varphi\nor(\psi\vee\sigma)\not\equiv(\varphi\nor\psi)\vee(\varphi\nor\sigma)$.}
b) follows from a) with the observation that $\bignor[2^{|\varphi|}]{i=1}\,\psi_i$ is satisfiable if and only if there is an $i\in\{1,\dots,2^{|\varphi|}\}$ such that $\psi_i$ is satisfiable. Note that given $i\in\{1,\dots,2^{|\varphi|}\}$ the formula $\psi_i$ can be computed from the original formula $\varphi$ in polynomial time by choosing (for all $j\in\{1,\dots,|\varphi|\}$) from the $j$th subformula of the form $\psi \nor \sigma$ the formula $\psi$ if the $j$th bit of $i$ is 0 and $\sigma$ if it is 1.
\end{proof}

We need the following simple property of monotone \MDL formulas.

\begin{lemma}\label{dep without negation}
Let $M$ be a set of \MDL operators with $\aneg\notin M$. Then an arbitrary $\MDL[M ]$ formula $\varphi$ is satisfiable iff the formula generated from $\varphi$ by replacing every dependence atom and every atomic proposition with the same atomic proposition $t$ is satisfiable.
\end{lemma}
\begin{proof}
If a frame $W$ is a model for $\varphi$, so is the frame generated from $W$ by setting all atomic propositions in all worlds to true.
\end{proof}

We are now able to classify some cases that can be easily reduced to known results.

\begin{corollary}\label{simple cases}
\begin{enumerate}[a)]
 \item If $\{\Box,\Diamond,\wedge,\vee,\aneg\}\subseteq M\subseteq\{\Box,\Diamond,\wedge,\vee,\aneg,\true,\false,\nor\}$ then \MDLSAT[M ] is \PSPACE-complete.
 \item If $\{\Box,\Diamond,\wedge,\vee,\false\}\subseteq M\subseteq\{\Box,\Diamond,\wedge,\vee,\true,\false,\dep{},\nor\}$ then \MDLSAT[M ] and \MDLSATk[M ] are \PSPACE-complete for all $k\geq 0$.
 \item If $\{\Box,\Diamond,\wedge,\false\}\subseteq M\subseteq\{\Box,\Diamond,\wedge,\true,\false,\dep{}\}$ then \MDLSAT[M ] and \MDLSATk[M ] are \coNP-complete  for all $k\geq 0$.
 \item If $M \subseteq \{\Box,\Diamond,\wedge,\vee,\true,\dep{},\nor\}$ then every \MDL[M ] formula is satisfiable.
 \item If $M \subseteq \{\wedge,\aneg, \true,\false,\dep{}\}$ then \MDLSAT[M ] is in \PTIME.
 \item If $M \subseteq \{\wedge,\vee, \true,\false,\dep{},\nor\}$ then \MDLSAT[M ] is in \PTIME.
\end{enumerate}
\end{corollary}
\begin{proof}
The lower bound of a) was shown by Ladner \cite{la77}, who proves \PSPACE-completeness for the case of full ordinary modal logic. The upper bound follows from this, Lemma~\ref{bullet distributivity} and $\existOperator\PSPACE=\PSPACE$.
The lower bound for b) was shown by Hemaspaandra \cite[Theorem~6.5]{he01} and the upper bound follows from a) together with Lemma~\ref{dep without negation}.

The lower bound for c) was shown by Donini et al.~\cite{dolenahonuma92} who prove \NP-hardness of the problem to decide whether an $\mathcal{ALE}$-concept is unsatisfiable. $\mathcal{ALE}$ is a description logic which essentially is nothing else then $\MDL[\Box,\Diamond,\wedge,\aneg,\true,\false]$ ($\aneg$ and $\true$ are not used in the hardness proof). For the upper bound Ladner's \PSPACE-algorithm \cite{la77} can be used, as in the case without disjunction it is in fact a \coNP-algorithm, together with Lemma~\ref{dep without negation}.

d) follows from Lemma~\ref{dep without negation} together with the fact that every \MDL formula with $t$ as the only atomic subformula is satisfied in the transitive singleton, i.e.~the frame consisting of only one state which has itself as successor, in which $t$ is true.

e) follows from the polynomial time complexity of deciding satisfiability of a \onecnf formula.
f) reduces to Boolean formula evaluation by Lemma~\ref{dep without negation}.
Note that for e) and f) dependence atoms can be replaced by \true{} because there we do not have any modality.
\end{proof}

\subsection{Poor man's dependence logic}

We now turn to the $\SigmaTwo$-complete cases. These include monotone poor man's logic, with and without dependence atoms.

\begin{theorem}\label{poor man bullet complexity}
If $\{\Box,\Diamond,\wedge,\aneg,\nor\} \subseteq M \subseteq\{\Box,\Diamond,\wedge,\aneg,\true,\false,\allowbreak\nor\}$ or $\{\Box,\Diamond,\wedge,\false,\nor\} \subseteq M\subseteq\{\Box,\Diamond,\wedge,\allowbreak\true,\allowbreak\false,\allowbreak\dep{},\allowbreak\nor\}$ then \MDLSAT[M ] and \MDLSATk[M ] are \SigmaTwo-complete for all $k\geq 0$.
\end{theorem}
\begin{proof}
Proving the upper bound for the second case reduces to proving the upper bound for the first case by Lemma~\ref{dep without negation}. For the first case it holds with Lemma~\ref{bullet distributivity} that $\MDLSAT[\Box,\Diamond,\wedge,\aneg,\true,\false,\nor]\in \existOperator\coNP=\SigmaTwo$ since $\MDLSAT[\Box,\Diamond,\wedge,\aneg,\true,\false]\in \coNP$. The latter follows directly from Ladner's \PSPACE-algo\-rithm for modal logic satisfiability \cite{la77} which is in fact a \coNP-algorithm in the case without disjunction.

For the lower bound we consider the quantified constraint satisfaction problem \SpecialQBF shown to be \PiTwo-complete by Bauland et al.~\cite{babocrrescvo10}. This problem can be reduced to the complement of $\MDLSAT[\Box,\Diamond,\wedge,\aneg/\false,\nor]$ in polynomial time.

An instance of \SpecialQBF consists of universally quantified Boolean variables $p_1,\dots,\allowbreak p_k$, existentially quantified Boolean variables $p_{k+1},\dots,p_n$ and a set of clauses each consisting of exactly three of those variables. \SpecialQBF is the set of all those instances for which for every truth assignment for $p_1,\dots,p_k$ there is a truth assignment for $p_{k+1},\dots,p_n$ such that in each clause exactly one variable evaluates to true.\footnote{For our reduction it is necessary that in each clause the variables are pairwise different whereas in \SpecialQBF this need not be the case. However, the \PiTwo-hardness proof can easily be adapted to account for this.}

For the reduction from \SpecialQBF to the complement of $\MDLSAT[\Box,\Diamond,\allowbreak\wedge,\allowbreak\aneg/\false,\allowbreak\nor]$ we extend a technique from the \coNP-hardness proof for $\MDLSAT[\Box,\allowbreak\Diamond,\wedge,\false]$ by Donini et al.~\cite[Theorem~3.3]{dolenahonuma92}. Let $p_1,\dots,p_k$ be the universally quantified and $p_{k+1},\dots,p_n$ the existentially quantified variables of a \SpecialQBF instance and let $C_1,\dots,C_m$ be its clauses (we assume w.l.o.g.~that each variable occurs in at least one clause). Then the corresponding $\MDL[\Box,\Diamond,\wedge,\allowbreak\false,\allowbreak\nor]$ formula is
\[\vspace{-0.3ex}\begin{array}{lcclllll}
\varphi:=&&\bigwedge\limits_{i=1}^k\;\big(&&\nabla_{i1}\dots\nabla_{im}&\nabla_{i1}\dots\nabla_{im}&\Box^{i-1}\Diamond\Box^{k-i}& p\\
&&&\nor&\Box^{m}&\Box^{m}&\Box^{i-1}\Diamond\Box^{k-i}& p\,\big)\\
&\wedge&\bigwedge\limits_{i=k+1}^n&&\nabla_{i1}\dots\nabla_{im}&\nabla_{i1}\dots\nabla_{im}&\Box^k&p\\
&\wedge&&&\Box^m&\Box^m&\Box^k&\false
\end{array}\vspace{-0.3ex}\]
where $p$ is an arbitrary atomic proposition and $\nabla_{ij} := \left\{\begin{array}{l@{\quad}l}\Diamond&\text{if $p_i \in C_j$}\\\Box&\text{else}\end{array}\right.$.

For the corresponding $\MDL[\Box,\Diamond,\wedge,\aneg,\nor]$ formula replace every $\false$ with $\neg p$.

To prove the correctness of our reduction we will need two claims.

\noindent\textbf{Claim~1.} For $r,s\geq0$ a $\MDL[\Box,\Diamond, \wedge,\aneg,\true, \false]$ formula $\Diamond\varphi_1\wedge\dots\wedge\Diamond\varphi_r\wedge\Box\psi_1\wedge\dots\wedge\Box\psi_s$ is unsatisfiable iff there is an $i\in\{1,\dots,r\}$ such that $\varphi_i\wedge\psi_1\wedge\dots\wedge\psi_s$ is unsatisfiable.

\noindent\emph{Proof of Claim~1.}
``$\Leftarrow$'': If $\varphi_i\wedge\psi_1\wedge\dots\wedge\psi_s$ is unsatisfiable, so is $\Diamond\varphi_i\wedge\Box\psi_1\wedge\dots\wedge\Box\psi_s$ and even more $\Diamond\varphi_1\wedge\dots\wedge\Diamond\varphi_r\wedge\Box\psi_1\wedge\dots\wedge\Box\psi_s$.

``$\Rightarrow$: Suppose that $\varphi_i\wedge\psi_1\wedge\dots\wedge\psi_s$ is satisfiable for all $i\in\{1,\dots,r\}$. Then $\Diamond\varphi_1\wedge\dots\wedge\Diamond\varphi_r\wedge\Box\psi_1\wedge\dots\wedge\Box\psi_s$ is satisfiable in a frame that consists of a root state and for each $i\in\{1,\dots,r\}$ a separate branch, reachable from the root in one step, which satisfies $\varphi_i\wedge\psi_1\wedge\dots\wedge\psi_s$.
\hfill$<<$

Note that $\Diamond\varphi_1\wedge\dots\wedge\Diamond\varphi_r\wedge\Box\psi_1\wedge\dots\wedge\Box\psi_s$ is always satisfiable if $r=0$.

\noindent\textbf{Definition.}
Let $v:\{p_1,\dots,p_k\}\to\{0,1\}$ be a valuation of $\{p_1,\dots,p_k\}$. Then $\varphi_v$ denotes the $\MDL[\Box,\Diamond,\wedge,\aneg/\false]$ formula
\[\begin{array}{cclllc}
&\bigwedge\limits_{\substack{i\in\{1,\dots,k\},\\v(p_i)=1}}&\nabla_{i1}\dots\nabla_{im}&\nabla_{i1}\dots\nabla_{im}&\Box^{i-1}\Diamond\Box^{k-i}& p\\
\wedge&\bigwedge\limits_{\substack{i\in\{1,\dots,k\},\\v(p_i)=0}}&\Box^{m}&\Box^{m}&\Box^{i-1}\Diamond\Box^{k-i}& p\\
\wedge&\bigwedge\limits_{i=k+1}^n&\nabla_{i1}\dots\nabla_{im}&\nabla_{i1}\dots\nabla_{im}&\Box^k&p\\
\wedge&&\Box^m&\Box^m&\Box^k&\neg p\,/\,\false
\end{array}\]

\noindent\textbf{Claim~2.}
Let $v:\{p_1,\dots,p_k\}\to\{0,1\}$ be a valuation. Then $\varphi_v$ is unsatisfiable iff $v$ can be continued to a valuation $v':\{p_1,\dots,p_n\}\to\{0,1\}$ such that in each of the clauses $\{C_1,\dots,C_m\}$ exactly one variable evaluates to true under $v'$.

\noindent\emph{Proof of Claim~2.}
By iterated use of Claim~1, $\varphi_v$ is unsatisfiable iff there are $i_1,\dots,i_{2m}$ with
\[\begin{array}{lcl}
i_j\in&&\big\{i\in\{1,\dots,n\}\mid \nabla_{ij'}=\Diamond\big\}\setminus\big\{i\in\{1,\dots,k\}\mid v(p_i)=0\big\}\\
&=&\big\{i\in\{1,\dots,n\}\mid p_i\in C_{j'}\big\}\setminus\big\{i\in\{1,\dots,k\}\mid v(p_i)=0\big\},
\end{array}\]
where $j':=
\left\{\begin{array}{l@{\;\;}l}j&\text{if $j\leq m$}\\j-m&\text{else}\end{array}\right.$,\quad
such that
\[\begin{array}{lcclc}
\varphi_v(i_1,\dots,i_{2m}):=&&\bigwedge\limits_{\substack{i\in\{1,\dots,k\},\\i\in\{i_1,\dots,i_{2m}\},\\v(p_i)=1}}&\Box^{i-1}\Diamond\Box^{k-i}& p\\
&\wedge&\bigwedge\limits_{\substack{i\in\{1,\dots,k\},\\v(p_i)=0}}&\Box^{i-1}\Diamond\Box^{k-i}& p\\
&\wedge&\bigwedge\limits_{\substack{i\in\{k+1,\dots,n\},\\i\in\{i_1,\dots,i_{2m}\}}}&\Box^k&p\\
&\wedge&&\Box^k&\neg p\,/\,\false
\end{array}\]
is unsatisfiable $(i)$ and such that there are no $a,b\in\{1,\dots,2m\}$ with $a<b$, $\nabla_{i_ba'}=\nabla_{i_bb'}=\Diamond$ (this is the case iff $p_{i_b}\in C_{a'}$ and $p_{i_b}\in C_{b'}$) and $i_a \neq i_b$ $(ii)$. The latter condition is already implied by Claim~1 as it simply ensures that no subformula is selected after it has already been discarded in an earlier step.
Note that $\varphi_v(i_1,\dots,i_{2m})$ is unsatisfiable iff for all $i\in\{1,\dots,k\}$: $v(p_i)=1$ and $i\in\{i_1,\dots,i_{2m}\}$\quad or\quad $v(p_i)=0$ (and $i\notin\{i_1,\dots,i_{2m}\}$) $(i')$.

We are now able to prove the claim.

``$\Leftarrow$'': For $j=1,\dots,2m$ choose $i_j\in\{1,\dots,n\}$ such that $p_{i_j}\in C_{j'}$ and $v'(p_{i_j})=1$. By assumption, all $i_j$ exist and are uniquely determined. Hence, for all $i\in\{1,\dots,k\}$ we have that $v(p_i)=0$ (and then $i\notin\{i_1,\dots,i_{2m}\}$) or $v(p_i)=1$ and there is a $j$ such that $i_j=i$ (because each variable occurs in at least one clause). Therefore condition $(i')$ is satisfied. Now suppose there are $a<b$ that violate condition $(ii)$. By definition of $i_b$ it holds that $p_{i_b}\in C_{b'}$ and $v'(p_{i_b})=1$. Analogously, $p_{i_a}\in C_{a'}$ and $v'(p_{i_a})=1$. By the supposition $p_{i_b}\in C_{a'}$ and $p_{i_a}\neq p_{i_b}$. But since $v'(p_{i_a})=v'(p_{i_b})=1$, that is a contradiction to the fact that in clause $C_{a'}$ only one variable evaluates to true.

``$\Rightarrow$'': If $\varphi_v$ is unsatisfiable, there are $i_1,\dots,i_{2m}$ such that $(i')$ and $(ii)$ hold. Let the valuation $v':\{p_1,\dots,p_n\}\to\{0,1\}$ be defined by
\[v'(p_i):=\left\{\begin{array}{l@{\ }l}1&\text{if $i\in\{i_1,\dots,i_{2m}\}$}\\0&\text{else}\end{array}\right..\]
Note that $v'$ is a continuation of $v$ because $(i')$ holds.

We will now prove that in each of the clauses $C_1,\dots,C_m$ exactly one variable evaluates to true under $v'$. Therefore let $j\in\{1,\dots,m\}$ be arbitrarily chosen.

By choice of $i_j$ it holds that $p_{i_j}\in C_j$. It follows by definition of $v'$ that $v'(p_{i_j})=1$. Hence, there is at least one variable in $C_j$ that evaluates to true.

Now suppose that besides $p_{i_j}$ another variable in $C_j$ evaluates to true. Then by definition of $v'$ it follows that there is a $\ell\in\{1,\dots,2m\}$, $\ell\neq j$, such that this other variable is $p_{i_\ell}$. We now consider two cases.

\emph{Case~$j < \ell$}:
This is a contradiction to $(ii)$ since, by definition of $\ell$, $p_{i_\ell}$ is in $C_{j'}$ as well as, by definition of $i_\ell$, in $C_{\ell'}$ and $i_j\neq i_\ell$.

\emph{Case~$\ell< j$}:
Since $j\in\{1,\dots,m\}$ it follows that $\ell\leq m$. Since $C_{\ell'}=C_{(\ell+m)'}$ it holds that $p_{i_{\ell+m}}\in C_{\ell'}$ and $p_{i_{\ell+m}}\in C_{(\ell+m)'}$. Furthermore $\ell<\ell+m$ and thus, by condition $(ii)$, it must hold that $i_\ell= i_{\ell+m}$. Therefore $p_{i_{\ell+m}}\in C_j$ and $v'(p_{i_{\ell+m}})=1$. Because $j<\ell+m$ this is a contradiction to condition $(ii)$ as in the first case.
\hfill$<<$

The correctness of the reduction now follows with the observation that $\varphi$ is equivalent to\\
$\bignordisplay{v:{\{p_1,\dots,p_k\}}\to\{0,1\}} \varphi_v$ and that $\varphi$ is unsatisfiable iff $\varphi_v$ is unsatisfiable for all valuations $v:{\{p_1,\dots,p_k\}}\to\{0,1\}$.

The \SpecialQBF instance is true iff every valuation $v:\{p_1,\dots,p_k\}\to \{0,1\}$ can be continued to a valuation $v':\{p_1,\dots,p_n\}\to\{0,1\}$ such that in each of the clauses $\{C_1,\dots,C_m\}$ exactly one variable evaluates to true under $v'$ iff, by Claim~2, $\varphi_{v}$ is unsatisfiable for all $v:\{p_1,\dots,p_k\}\to\{0,1\}$ iff, by the above observation, $\varphi$ is unsatisfiable.
\end{proof}

Next we turn to (non-monotone) poor man's logic.

\begin{theorem}\label{poor man dep complexity}
If $\{\Box,\Diamond,\wedge,\aneg,\dep{}\}\subseteq M$ then \MDLSAT[M ] is \NEXPTIME-complete.
\end{theorem}
\begin{proof}
Sevenster showed that the problem is in \NEXPTIME in the case of $\nor\notin M$ \cite[Lemma~14]{se09}. Together with Lemma~\ref{bullet distributivity} and the fact that $\existOperator\NEXPTIME=\NEXPTIME$ the upper bound applies.

For the lower bound we reduce \DQBFCNF, which was shown to be \NEXPTIME-hard by Peterson et al.~\cite[Lemma~5.2.2]{pereaz01}\footnote{Peterson et al.~showed \NEXPTIME-hardness for \dqbf without the restriction that the formulae must be in \cnf. However, the restriction does not lower the complexity since every propositional formula is satisfiability-equivalent to a formula in \cnf whose size is bounded by a polynomial in the size of the original formula.}, to our problem.

An instance of \DQBFCNF consists of universally quantified Boolean variables $p_1,\allowbreak \dots,\allowbreak p_k$, existentially quantified Boolean variables $p_{k+1},\dots,p_n$, dependence constraints $P_{k+1},\allowbreak \dots,\allowbreak P_n\subseteq\{p_1,\dots,p_k\}$ and a set of clauses each consisting of three (not necessarily distinct) literals. Here, $P_i$ intuitively states that the value of $p_i$ only depends on the values of the variables in $P_i$. Now, \DQBFCNF is the set of all those instances for which there is a collection of functions $f_{k+1},\dots,f_n$ with $f_i:\{0,1\}^{P_i}\to\{0,1\}$ such that for every valuation $v:\{p_1,\dots,p_k\}\to\{0,1\}$ there is at least one literal in each clause that evaluates to true under the valuation $v':\{p_1,\dots,p_n\}\to\{0,1\}$ defined by
\[v'(p_i):=\left\{\begin{array}{l@{\quad}l}
v(p_i)&\text{if $i\in\{1,\dots,k\}$}\\
f_i(v\upharpoonright P_i)&\text{if $i\in\{k+1,\dots,n\}$}
\end{array}\right.\enspace.\]

The functions $f_{k+1},\dots,f_n$ act as restricted existential quantifiers, i.e., for an $i\in\{k+1,\dots,n\}$ the variable $p_i$ can be assumed to be existentially quantified dependent on all universally quantified variables in $P_i$ (and, more importantly, independent of all universally quantified variables not in $P_i$).
Dependencies are thus explicitly specified through the dependence constraints and can contain -- but are not limited to -- the traditional sequential dependencies, e.g.~the quantifier sequence $\forall p_1 \exists p_2 \forall p_3 \exists p_4$ can be modeled by the dependence constraints $P_2=\{p_1\}$ and $P_4=\{p_1,p_3\}$.

For the reduction from \DQBFCNF to $\MDLSAT[\Box,\Diamond,\allowbreak\wedge,\allowbreak\aneg,\dep{}]$ we use an idea from Hemaspaandra \cite[Theorem~4.2]{he01}. There, \PSPACE-hardness of $\MDLSAT[\Box,\Diamond,\allowbreak\wedge,\allowbreak\aneg]$ over the class $\mathcal F_{\leq 2}$ of all Kripke structures in which every world has at most two successors is shown. The crucial point in the proof is to ensure that every Kripke structure satisfying the constructed $\MDL[\Box,\Diamond,\wedge,\aneg]$ formula adheres to the structure of a complete binary tree and does not contain anything more than this tree. In the class $\mathcal F_{\leq 2}$ this is automatically the case since in a complete binary tree all worlds already have two successors.

Although in our case there is no such a priori restriction and therefore we cannot make sure that every satisfying structure is not more than a binary tree, we are able to use dependence atoms to ensure that everything in the structure that does not belong to the tree is essentially nothing else than a copy of a subtree. This will be enough to show the desired reducibility.

Let $p_1,\dots,p_k$ be the universally quantified and $p_{k+1},\dots,p_n$ the existentially quantified variables of a \DQBFCNF instance $\varphi$ and let $P_{k+1},\dots,P_n$ be its dependence constraints and $\{l_{11},l_{12},l_{13}\},\dots,\allowbreak\{l_{m1},\allowbreak l_{m2},l_{m3}\}$ its clauses. Then the corresponding $\MDL[\Box,\Diamond,\allowbreak\wedge,\allowbreak\aneg,\dep{}]$ formula is
\[\begin{array}{lcl@{\quad}c}
g(\varphi):=&&\bigwedge\limits_{i=1}^n \Box^{i-1}(\Diamond \Box^{n-i}p_i \wedge \Diamond \Box^{n-i}\overline{p_i})&(i)\\
&\wedge&\bigwedge\limits_{i=1}^m \Diamond^n (\overline{l_{i1}} \wedge \overline{l_{i2}} \wedge \overline{l_{i3}} \wedge f_i)&(ii)\\
&\wedge&\bigwedge\limits_{i=1}^m \Box^n \dep[l'_{i1},l'_{i2},l'_{i3}]{f_i}&(iii)\\
&\wedge&\Box^k\Diamond^{n-k}\big(\overline{f_1} \wedge \dots \wedge \overline{f_m} \;\wedge\; \bigwedge_{i=k+1}^n \dep[P_i]{p_i}\big)&(iv)
\end{array}\]
where $p_1,\dots,p_n,f_1,\dots,f_m$ are atomic propositions and $l'_{ij}:=\left\{\begin{array}{l@{\quad}l}p&\text{if $l_{ij}=p$}\\p&\text{if $l_{ij}=\overline{p}$}\end{array}\right.\enspace$.

Now if $\varphi$ is valid, consider the frame which consists of a complete binary tree with $n$ levels (not counting the root) and where each of the $2^n$ possible labelings of the atomic propositions $p_1,\dots,p_n$ occurs in exactly one leaf. Additionally, for each $i\in\{1,\dots, m\}$ $f_i$ is labeled in exactly those leaves in which $l_{i1}\vee l_{i2}\vee l_{i3}$ is false. This frame obviously satisfies $(i)$, $(ii)$ and $(iii)$. And since the modalities in $(iv)$ model the quantors of $\varphi$, $\overline{f_i}$ is true exactly in the leaves in which $l_{i1}\vee l_{i2}\vee l_{i3}$ is true and the \dep{} atoms in $(iv)$ model the dependence constraints of $\varphi$, $(iv)$ is also true and therefore $g(\varphi)$ is satisfied in the root of the tree.

As an example see Fig.~\ref{fig:generated frame} for a frame satisfying $g(\varphi)$ if the first clause in $\varphi$ is $\{\overline{p_1},p_n\}$.
\begin{figure}[ht]
\begin{center}
\begin{tikzpicture}[auto, every state/.style={minimum size=2em}]
\node[state]           (z0)  at (12,10)                   {};
\node[state]           (z11) at (8,9)                     {};
\node                  (z12) at (14,8)                   {{\Large $\vdots$}};
\node[state]           (zn1) at (6,8)                     {};
\node[state,label=below:{$\begin{array}{c}p_1\\p_2\\\vdots\\p_n\end{array}$}]           (zn1a) at (5,7)                     {};
\node[state,label={below:{$\begin{array}{c}p_1\\p_2\\\vdots\\\overline{p_n}\\f_1\end{array}$}}] (zn1b) at (7,7)                     {};
\node                  (znn1) at (8,7)                    {$\cdots$};
\node[state]           (zn2) at (10,8)                     {};
\node[state,label=below:{$\begin{array}{c}p_1\\\overline{p_2}\\\vdots\\p_n\end{array}$}]           (zn2a) at (9,7)                     {};
\node[state,label={below:{$\begin{array}{c}p_1\\\overline{p_2}\\\vdots\\\overline{p_n}\\f_1\end{array}$}}] (zn2b) at (11,7)                     {};
\node                  (znn2) at (12,7)                    {{\Large $\cdots$}};

\path[->]       (z0)    edge                    node {$p_1$}        (z11)
                (zn1)   edge                    node {$p_n$}        (zn1a)
                (zn1)   edge                    node {$\overline{p_n}$}        (zn1b)
                (zn2)   edge                    node {$p_n$}        (zn2a)
                (zn2)   edge                    node {$\overline{p_n}$}        (zn2b);
\path[->,dotted] (z11)   edge                    node {$p_2$}        (zn1)
                (z11)   edge                    node {$\overline{p_2}$}        (zn2)
                (z0)    edge                    node {$\overline{p_1}$}        (z12);
\end{tikzpicture}
\caption{Frame satisfying $g(\varphi)$}
\label{fig:generated frame}
\end{center}
\end{figure}
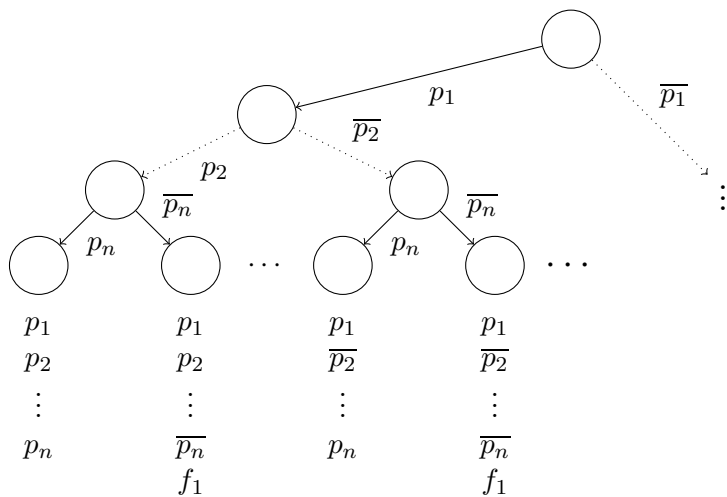

If, on the other hand, $g(\varphi)$ is satisfiable, let $W$ be a frame and $t$ a world in $W$ such that $W,\{t\}\models g(\varphi)$.
Now $(i)$ enforces $W$ to contain a complete binary tree $T$ with root $t$ such that each labeling of $p_1,\dots,p_n$ occurs in a leaf of $T$.

We can further assume w.l.o.g.~that $W$ itself is a tree since in \MDL different worlds with identical proposition labelings are indistinguishable and therefore every frame can simply be unwinded to become a tree. Since the modal depth of  $g(\varphi)$ is $n$ we can assume that the depth of $W$ is at most $n$. And since $(i)$ enforces that every path in $W$ from $t$ to a leaf has a length of at least $n$, all leaves of $W$ lie at levels greater or equal to $n$. Altogether we can assume that $W$ is a tree, that all its leaves lie at level $n$ and that it has the same root as $T$. The only difference is that the degree of $W$ may be greater than that of $T$.

But we can nonetheless assume that up to level $k$ the degree of $W$ is~2 $(*)$. This is the case because if any world up to level $k-1$ had more successors than the two lying in $T$, the additional successors could be omitted and $(i)$, $(ii)$, $(iii)$ and $(iv)$ would still be fulfilled. For $(i)$, $(ii)$ and $(iii)$ this is clear and for $(iv)$ it holds because $(iv)$ begins with $\Box^k$.

We will now show that, although $T$ may be a proper subframe of $W$, $T$ is already sufficient to fulfill $g(\varphi)$. From this the validity of $\varphi$ will follow immediately.

\noindent\textbf{Claim.}
$T,\{t\}\models g(\varphi)$.

\noindent\emph{Proof of Claim.}
We consider sets of leaves of $W$ that satisfy $\overline{f_1} \wedge \dots \wedge \allowbreak \overline{f_m} \;\wedge\;\allowbreak \bigwedge_{i=k+1}^n \dep[P_i]{p_i}$ and that can be reached from the set $\{t\}$ by the modality sequence $\Box^k\Diamond^{n-k}$. Let $S$ be such a set and let $S$ be chosen so that there is no other such set that contains less worlds outside of $T$ than $S$ does. Assume there is a $s\in S$ that does not lie in $T$.

Let $i\in\{1,\dots,m\}$ and let $s'$ be the leaf in $T$ that agrees with $s$ on the labeling of $p_1,\dots,p_n$. Then, with $W,\{s\}\models \overline{f_i}$ and $(iii)$, it follows that $W,\{s'\}\models\overline{f_i}$.

Let $S' := (S\setminus\{s\})\cup\{s'\}$. Then it follows by the previous paragraph that $W,S'\models \overline{f_1} \wedge \dots \wedge \overline{f_m} $. Since $W,S\models\bigwedge_{i=k+1}^n \dep[P_i]{p_i}$ and $s'$ agrees with $s$ on the propositions $p_1,\dots,p_n$ it follows that $W,S'\models\bigwedge_{i=k+1}^n \dep[P_i]{p_i}$. Hence, $S'$ satisfies $\overline{f_1} \wedge \dots \wedge \overline{f_m} \;\wedge\; \bigwedge_{i=k+1}^n \dep[P_i]{p_i}$ and as it only differs from $S$ by replacing $s$ with $s'$ it can be reached from $\{t\}$ by $\Box^k\Diamond^{n-k}$ because $s$ and $s'$ agree on $p_1,\dots,p_k$ and, by $(*)$, $W$ does not differ from $T$ up to level $k$. But this is a contradiction to the assumption since $S'$ contains one world less than $S$ outside of $T$. Thus, there is no $s\in S$ that does not lie in $T$ and therefore $(iv)$ is fulfilled in $T$. Since $(i)$, $(ii)$ and $(iii)$ are obviously also fulfilled in $T$, it follows that $T,\{t\}\models g(\varphi)$.
\hfill$<<$

$(ii)$ ensures that for all $i\in\{1,\dots,m\}$ there is a leaf in $W$ in which $\neg(l_{i1}\vee l_{i2}\vee l_{i3})\wedge f_i$ is true. This leaf can lie outside of $T$. However, $(iii)$ ensures that all leaves that agree on the labeling of $l_{i1}$, $l_{i2}$ and $l_{i3}$ also agree on the labeling of $f_i$. And since there is a leaf where $\neg(l_{i1}\vee l_{i2}\vee l_{i3})\wedge f_i$ is true, it follows that in all leaves, in which $\neg(l_{i1}\vee l_{i2}\vee l_{i3})$ is true, $f_i$ is true. Conversely, if $\overline{f_i}$ is true in an arbitrary leaf of $W$ then so is $l_{i1}\vee l_{i2}\vee l_{i3}$ $(**)$.

The modality sequence $\Box^k\Diamond^{n-k}$ models the quantors of $\varphi$ and $\bigwedge_{i=k+1}^n \dep[P_i]{\allowbreak p_i}$ models its dependence constraints. And so there is a bijective correspondence between sets of worlds reachable in $T$ by $\Box^k\Diamond^{n-k}$ from $\{t\}$ and that satisfy $\bigwedge_{i=k+1}^n \dep[P_i]{p_i}$ on the one hand and truth assignments to $p_1,\dots,p_n$ generated by the quantors of $\varphi$ and satisfying its dependence constraints on the other hand.
Additionally, by $(**)$ follows that $\overline{f_1}\wedge\dots\wedge\overline{f_m}$ implies $\bigwedge_{i=1}^m (l_{i1}\vee l_{i2} \vee l_{i3})$ and since $T,\{t\}\models g(\varphi)$, $\varphi$ is valid.
\end{proof}

\subsection{Cases with only one modality}

We finally examine formulas with only one modality.

\begin{theorem}\label{only one modality}
Let $M \subseteq \{\Box, \Diamond, \wedge, \vee, \aneg, \true, \false, \nor\}$ with $\Box\notin M$ or $\Diamond \notin M$. Then the following hold:
\begin{enumerate}[a)]
 \item \MDLSAT[M \cup \{\dep{}\}] $\leqpm$ \MDLSAT[M \cup \{\true, \false\}], i.e.,~adding the $\dep{}$ operator does not increase the complexity if we only have one modality.
  \item For every \MDL[M \cup\{\dep{}\}] formula $\varphi$ it holds that $\nor$ is equivalent to $\vee$, i.e.,~$\varphi$ is equivalent to every formula that is generated from $\varphi$ by replacing some or all occurrences of $\nor$ by $\vee$ and vice versa.
\end{enumerate}
\end{theorem}
\begin{proof}
Every negation $\neg \dep{}$ of a dependence atom is by definition always equivalent to $\false$ and can thus be replaced by the latter. For positive $\dep{}$ atoms and the $\nor$ operator we consider two cases.

\emph{Case $\Diamond\notin M$.} If an arbitrary \MDL[\Box,\wedge, \vee, \aneg, \true, \false, \dep{}, \nor] formula $\varphi$ is satisfiable then it is so in an intransitive singleton frame, i.e.~a frame that only contains one world which does not have a successor, because there every subformula that begins with a $\Box$ is automatically satisfied. In a singleton frame all $\dep{}$ atoms obviously hold and $\nor$ is equivalent to $\vee$. Therefore the (un-)satisfiability of $\varphi$ is preserved when substituting every $\dep{}$ atom in $\varphi$ with $\true$ and every $\nor$ with $\vee$ (or vice versa).

\emph{Case $\Box \notin M$.} If an arbitrary \MDL[\Diamond,\wedge, \vee, \aneg, \true, \false, \dep{}, \nor] formula $\varphi$ is satisfiable then, by the downward closure property, there is a frame $W$ with a world $s$ such that $W,\{s\} \models \varphi$. Since there is no $\Box$ in $\varphi$, every subformula of $\varphi$ is also evaluated in a singleton set (because a $\Diamond$ can never increase the cardinality of the evaluation set). And as in the former case we can replace every $\dep{}$ atom with $\true$ and every $\nor$ with $\vee$ (or vice versa).
\end{proof}

Thus we obtain the following consequences -- note that with the preceding results this takes care of all cases in Table~\ref{results}.
\begin{corollary}\label{one modality cases}
\begin{enumerate}[a)]
 \item If $\{\wedge,\aneg\}\subseteq M\subseteq\{\Box,\Diamond,\wedge,\vee,\aneg,\true,\false,\dep{},\nor\}$, $M\cap\{\vee,\nor\}\neq \emptyset$ and $|M\cap\{\Box,\Diamond\}|=1$ then \MDLSAT[M ] and \MDLSATk[M ] are \NP-complete for all $k\geq 0$.
 \item If $\{\wedge,\aneg\}\subseteq M\subseteq\{\Box,\Diamond,\wedge,\aneg,\true,\false,\dep{}\}$ and $|M\cap\{\Box,\Diamond\}|=1$ then $\MDLSAT[M ] \in\PTIME$.
 \item If $\{\wedge\}\subseteq M\subseteq\{\Box,\Diamond,\wedge,\vee,\true,\false,\dep{},\nor\}$ and $|M\cap\{\Box,\Diamond\}|=1$ then $\MDLSAT[M ] \in \PTIME$.
 \item If $\wedge \notin M$ then $\MDLSAT[M ] \in \PTIME$.
\end{enumerate}
\end{corollary}
\begin{proof}
a) without the $\dep{}$ and $\nor$ operators is exactly \cite[Theorem~6.2(2)]{he01}. Theorem~\ref{only one modality}a,b extends the result to the case with the new operators.
b) is \cite[Theorem~6.4(c,d)]{he01} together with Theorem~\ref{only one modality}a for the $\dep{}$ operator.
c) is \cite[Theorem~6.4(e,f)]{he01} together with Theorem~\ref{only one modality}a,b.

d) without $\dep{}$ and $\nor$ is \cite[Theorem~6.4(b)]{he01}. The proof for the case with the new operators is only slightly different: Let $\varphi$ be an arbitrary \MDL[M ] formula. By the same argument as in the proof of Theorem~\ref{only one modality}b we can replace all top-level (i.e.~not lying inside a modality) occurrences of $\nor$ in $\varphi$ with $\vee$ to get the equivalent formula $\varphi'$. $\varphi'$ is of the form $\Box \psi_1\vee\dots\vee\Box\psi_k\vee\Diamond\sigma_1\vee\dots\vee\Diamond\sigma_m\vee a_1\vee\dots\vee a_s$ where every $\psi_i$ and $\sigma_i$ is a \MDL[M ] formula and every $a_i$ is an atomic formula. If $k>0$ or any $a_i$ is a literal, $\true$ or a dependence atom then $\varphi'$ is satisfiable. Otherwise it is satisfiable iff one of the $\sigma_i$ is satisfiable and this can be checked recursively in polynomial time.
\end{proof}

\subsection{Bounded arity dependence}

\begin{theorem}\label{bounded dep upper}
Let $k\geq 0$. Then the following holds:
\begin{enumerate}[a)]
 \item If $M \subseteq \{\Box, \Diamond, \wedge, \vee, \aneg, \true, \false, \dep{}\}$ then $\MDLSATk[M ] \in \PSPACE$.
 \item If $M \subseteq \{\Box, \Diamond, \wedge, \aneg, \true, \false, \dep{}\}$ then $\MDLSATk[M ] \in \SigmaThree$.
\end{enumerate}
\end{theorem}
\begin{proof}
a) Let $\varphi \in \MDLk[M ]$. Then by \cite[Theorem~6]{se09} there is an ordinary modal logic formula $\varphi^T$ which is equivalent to $\varphi$ on singleton sets of evaluation, i.e., for all Kripke structures $W$ and states $w$ in $W$
\[W,\{w\}\models \varphi \text{\quad{}iff\quad} W,w\models \varphi^T.\]
Here $\varphi^T$ is constructed from $\varphi$ in the following way:
Let $\dep[p_{i_{1,1}},\dots,p_{i_{1,k_1}}]{p_{i_{1,k_1+1}}},\allowbreak\dots,\allowbreak\dep[p_{i_{n,1}},\dots,p_{i_{n,k_n}}]{p_{i_{n,k_n+1}}}$ be all dependence atoms occurring inside $\varphi$ (in an arbitrary order and including multiple occurrences of the same atom in $\varphi$ multiple times) and for all $j \geq 0$ let
\[B_j := \{\alpha_f(p_1,\dots,p_j) \mid f:\{\true,\false\}^j \to \{\true,\false\}\text{ is a total Boolean function}\},\]
where $\alpha_f(p_1,\dots,p_j)$ is the propositional encoding of $f$, i.e.,
\[\alpha_f(p_1,\dots,p_j) := \bigvee_{(i_1,\dots,i_j)\in f^-1(\true)}\, p_1^{i_1}\wedge\dots\wedge p_j^{i_j},\]
with $p^i:=\left\{\begin{array}{l@{\text{ if }}l}p&i=\true\\\neg p&i=\false\end{array}\right.$.
Note that for all $f:\{\true,\false\}^j \to \{\true,\false\}$ and all valuations $V:\{p_1,\dots,p_j\}\to \{\true,\false\}$ it holds that $V \models \alpha_f$ iff $f(V(p_1),\dots,V(p_j))=\true$.

Then $\varphi^T$ is defined as
\[\bigvee_{\alpha_1\in B_{k_1}} \dots \bigvee_{\alpha_n\in B_{k_n}}\,\varphi'(\alpha_1,\dots,\alpha_n),\]
where $\varphi'(\alpha_1,\dots,\alpha_n)$ is generated from $\varphi$ by replacing each dependence atom $\dep[p_{i_{\ell,1}},\allowbreak\dots,\allowbreak p_{i_{\ell,k_\ell}}]{p_{i_{\ell,k_\ell+1}}}$ with the propositional formula $\alpha_\ell(p_{i_{\ell,1}},\dots,p_{i_{\ell,k_{\ell}}}) \leftrightarrow p_{i_{\ell,k_\ell+1}}$.
Note that for all $\ell\in\{1,\dots,n\}$ we have that $|\alpha_\ell| \in \bigO{2^{k_\ell}}$ and $|B_{k_\ell}| = 2^{2^{k_\ell}}$. Therefore
\[|\varphi^T| \in \prod_{1\leq \ell \leq n} 2^{2^{k_\ell}}\,\cdot\,|\varphi|\cdot \bigO{2^{k_\ell}} \subseteq \bigO{\left(2^{2^k}\right)^n\cdot |\varphi|}.\]
This means that $\varphi^T$ is an exponentially (in the size of $\varphi$) large disjunction of terms of linear size. $\varphi^T$ is satisfiable if and only if at least one of its terms is satisfiable. Hence we can nondeterministically guess in polynomial time which one of the exponentially many terms should be satisfied and then check in deterministic polynomial space whether this one is satisfiable. The latter is possible because $\varphi'(\alpha_1,\dots,\alpha_n)$ is an ordinary modal logic formula and the satisfiability problem for this logic is in \PSPACE \cite{la77}. Altogether this leads to $\MDLSATk[M ] \in \existOperator \PSPACE = \PSPACE$.

\smallskip
b) In this case we cannot use the same argument as before without modifications since that would only lead to a \PSPACE upper bound again. The problem is that in the contruction of $\varphi^T$ we introduce the subformulas $\alpha_\ell$ and these may contain the $\vee$ operator. We can, however, salvage the construction by looking inside Ladner's \PSPACE algorithm \cite[Theorem~5.1]{la77}. For convenience we restate the algorithm in Listing~\ref{algo}. It holds for all ordinary modal logic formulas $\varphi$ that $\varphi$ is satisfiable if and only if \lstinline!satisfiable($\{\varphi\}$, $\emptyset$, $\emptyset$)!$=\true$.

\begin{lstlisting}[float=!ht,caption=Algorithm \lstinline!satisfiable($T$\, $A$\, $E$)!,label=algo]
if $T\nsubseteq \atom$ then
  choose $\psi\in T\setminus \atom$ //deterministically (but arbitrarily)
  set $T' := T\setminus \{\psi\}$
  if $\psi = \psi_1 \wedge \psi_2$ then
    return satisfiable($T' \cup \{\psi_1, \psi_2\}$, $A$, $E$)
  elseif $\psi = \psi_1 \vee \psi_2$ then
    nondeterministically existentially guess $i\in\{1,2\}$
    return satisfiable($T' \cup \{\psi_i\}$, $A$, $E$)
  elseif $\psi = \Box\psi_1$ then
    return satisfiable($T'$, $A\cup\{\psi_1\}$, $E$)
  elseif $\psi = \Diamond\psi_1$ then
    return satisfiable($T'$, $A$, $E\cup\{\psi_1\}$)
  end
else
  if $T$ is consistent then
    if $E \neq \emptyset$
      nondeterministically universally guess $\psi \in E$
      return satisfiable($A \cup \{\psi\}$, $\emptyset$, $\emptyset$)
    else
      return $\true$
    end
  else
    return $\false$
  end
end
$\textnormal{\parbox{\textwidth}{Here \atom denotes the set of atomic propositions, their negations and the constants \true{} and \false.}}$
\end{lstlisting}

The algorithm works in a top-down manner and runs in alternating polynomial time. It universally guesses when encountering a $\Box$ operator and existentially guesses when encountering a $\vee$ operator -- in all other cases it is deterministic.
Now, to check whether $\varphi^T$ is satisfiable we first existentially guess which of the exponentially many terms should be satisfied and then check whether this term $\varphi'(\alpha_1,\dots,\alpha_n)$ is satisfiable by invoking \lstinline!satisfiable($\{\varphi'(\alpha_1,\dots,\alpha_n)\}$, $\emptyset$, $\emptyset$)!.

To see that this in fact gives us a $\SigmaThree$-algorithm note that $\varphi$ does not contain any disjunctions. Hence also $\varphi'(\alpha_1,\dots,\alpha_n)$ contains no disjunctions apart from the ones that occur inside one of the subformulas $\alpha_1,\dots,\alpha_n$. Therefore the algorithm \lstinline!satisfiable! does not do any nondeterministic existential branching apart from when processing an $\alpha_i$. But in the latter case it is impossible to later nondeterministically universally branch because univeral guessing only occurs when processing a $\Box$ operator and these cannot occur inside an $\alpha_i$, since these are purely propositional formulas. Therefore the \lstinline!satisfiable! algorithm, if run on a formula $\varphi'(\alpha_1,\dots,\alpha_n)$ as input, is essentially a $\PiTwo$ algorithm. Together with the existential guessing of the term in the beginning we get that $\MDLSATk[M ] \in \existOperator \PiTwo = \SigmaThree$.
\end{proof}

\begin{theorem}\label{poor man bounded}
If $\{\Box, \Diamond, \wedge, \aneg, \dep{}\} \subseteq M$ then $\MDLSATpara[M ]{3}$ is \SigmaThree-hard.
\end{theorem}
\begin{proof}
We use the same construction as in the hardness proof for Theorem~\ref{poor man dep complexity} to reduce the problem \qbfcnf, which was shown to be \SigmaThree-complete by Wrathall~\cite[Corollary~6]{wr77}, to our problem. \qbfcnf is the set of all propositional sentences of the form
\[\exists p_1\dots\exists p_k \forall p_{k+1}\dots \forall p_\ell \exists p_{\ell+1}\dots \exists p_n \bigwedge_{i=1}^m (l_{i1}\vee l_{i2} \vee l_{i3}),\]
where the $l_{ij}$ are literals over $p_1,\dots,p_n$, which are valid.

Now let $\varphi$ be a \qbfcnf instance, let $p_1,\dots,p_n$ be its variables and let $k$, $\ell$, $m$, $(l_{ij})_{\substack{i=1,\dots,m\\j=1,2,3}}$ be as above. Then the corresponding $\MDLpara[\Box, \Diamond, \wedge, \aneg, \dep{}]{3}$ formula is
\[\begin{array}{lcl@{\quad}c}
g(\varphi):=&&\bigwedge\limits_{i=1}^n \Box^{i-1}(\Diamond \Box^{n-i}p_i \wedge \Diamond \Box^{n-i}\overline{p_i})&(i)\\
&\wedge&\bigwedge\limits_{i=1}^m \Diamond^n (\overline{l_{i1}} \wedge \overline{l_{i2}} \wedge \overline{l_{i3}} \wedge f_i)&(ii)\\
&\wedge&\bigwedge\limits_{i=1}^m \Box^n \dep[l'_{i1},l'_{i2},l'_{i3}]{f_i}&(iii)\\
&\wedge&\Diamond^k\Box^{\ell-k}\Diamond^{n-\ell}(\dep{p_1}\wedge \dots \wedge \dep{p_k}\ \wedge\ \overline{f_1} \wedge \dots \wedge \overline{f_m})&(iv)
\end{array}\]
where $p_1,\dots,p_n,f_1,\dots,f_m$ are atomic propositions and $l'_{ij}:=\left\{\begin{array}{l@{\quad}l}p&\text{if $l_{ij}=p$}\\p&\text{if $l_{ij}=\overline{p}$}\end{array}\right.\enspace$.

The proof that $g$ is a correct reduction is essentially the same as for Theorem~\ref{poor man dep complexity}. The only difference is that there we had arbitrary dependence atoms in part $(iv)$ of $g(\varphi)$ whereas here we only have 0-ary dependence atoms. This difference is due to the fact that there we had to be able to express arbitrary dependencies because we were reducing from \DQBFCNF whereas here we only have two kinds of dependencies for the existentially quantified variables: either complete constancy (for the variables that get quantified before any universal variables does) or complete freedom (for the variables that get quantified after all universal variables are already quantified). The former can be expressed by 0-ary dependence atoms and for the latter we simply omit any dependence atoms.

Note that it might seem as if with the same construction even $\Sigma^p_k$-hardness for arbitrary $k$ could be proved by having more alternations between the two modalities in part $(iv)$ of $g(\varphi)$. The reason that this does not work is that we do not really ensure that a structure fulfilling $g(\varphi)$ is not more than a binary tree, e.g.~it can happen that the root node of the tree has three successors: one in whose subtree all leaves on level $n$ are labeled with $p_1$, one in whose subtree no leaves are labeled with $p_1$ and one in whose subtree only some leaves are labeled with $p_1$. Now, the first diamond modality can branch into this third subtree and then the value of $p_1$ is not yet determined. Hence the modalities alone are not enough to express alternating dependencies and hence we need the $\dep{p_i}$ atoms in part $(iv)$ to ensure constancy.
\end{proof}

\begin{corollary}\label{bounded dep concrete}
\begin{enumerate}[a)]
 \item Let $k\geq 0$ and $\{\Box, \Diamond, \wedge, \vee, \aneg\} \subseteq M$. Then $\MDLSATk[M ]$ is \PSPACE-complete.
 \item Let $k\geq 3$ and $\{\Box, \Diamond, \wedge, \aneg, \dep{}\} \subseteq M \subseteq \{\Box, \Diamond, \wedge, \aneg, \true, \false, \dep{}, \nor\}$. Then $\MDLSATk[M ]$ is \SigmaThree-complete.
\end{enumerate}
\end{corollary}
\begin{proof}
The lower bound for a) is due to the \PSPACE-completeness of ordinary modal logic satisfiability which was shown in \cite{la77}. The upper bound follows from Theorem~\ref{bounded dep upper}a, Lemma~\ref{bullet distributivity}b and the fact that $\existOperator \PSPACE = \PSPACE$.

The lower bound for b) is Theorem~\ref{poor man bounded}. The upper bound follows from Theorem~\ref{bounded dep upper}b, Lemma~\ref{bullet distributivity} and $\existOperator \SigmaThree = \SigmaThree$.
\end{proof}

\section{Conclusion}

In this paper we completely classified the complexity of the
satisfiability problem for modal dependence logic for all fragments of the
language defined by restricting the modal and propositional operators to a
subset of those considered by V\"a\"an\"anen and Sevenster.
Our results show a dichotomy for the $\dep{}$ operator; either the complexity jumps to \NEXPTIME-completeness when introducing $\dep{}$ or it does not increase at all -- and in the latter case the $\dep{}$ operator does not increase the expressiveness of the logic.
Intuitively, the \NEXPTIME-completeness can be understood as the complexity of guessing Boolean functions of unbounded arity.

In an earlier version \cite{lovo10} of this paper we formulated the question whether there are natural fragments of modal dependence logic where adding the dependence operator does not let the complexity of satisfiability testing jump up to \NEXPTIME but still increases the expressiveness of the logic. We can now give an answer to that question; by restricting the arity of the \dep{} operator. In this case the dependence becomes too weak to increase the complexity beyond \PSPACE. However, in the case of poor man's logic, i.e.~only disjunctions are fobidden, the complexity increases to $\SigmaThree$ when introducing dependence but it still is not as worse as for full modal logic.
Intuitively,  the complexity drops below \NEXPTIME because the Boolean functions which have to be guessed are now of a bounded arity.

In a number of precursor papers, e.\,g., \cite{le79} on propositional
logic or \cite{hescsc10} on modal logic, not only subsets of the classical
operators $\{\Box,\Diamond,\wedge,\vee,\aneg\}$ were considered but also
propositional connectives given by arbitrary Boolean functions. The main
result of Lewis, e.\,g., can be succinctly summarized as follows:
Propositional satisfiability is $\NP$-complete if and only if in the input
formulas the connective $\varphi \wedge\neg \psi$ is allowed (or can be
``implemented'' with the allowed connectives).

We consider it interesting to initiate such a more general study for modal
dependence logic and determine the computational complexity of
satisfiability if the allowed connectives are taken from a fixed class in
Post's lattice. Contrary to propositional or modal logic, however, the
semantics of such generalized formulas is not clear a priori -- for
instance, how should exclusive-or be defined in dependence logic? Even for
simple implication, there seem to be several reasonable definitions,
cf.~\cite{abva08}.

A further possibly interesting restriction of dependence logic might be to
restrict the type of functional dependence beyond simply restricting the arity. Right now, dependence just
means that there is some function whatsoever that determines the value of
a variable from the given values of certain other variables. Also here it
might be interesting to restrict the function to be taken from a fixed
class in Post's lattice, e.\,g., to be monotone or self-dual.

Finally, it seems natural to investigate the possibility of enriching classical temporal logics as \LTL, \CTL or \CTLs with dependence as some of them are extensions of classical modal logic. The questions here are of the same kind as for \MDL: expressivity, complexity, fragments, etc.

\nocite{va08}
\nocite{hescsc10}
\nocite{he05}


\newcommand{\etalchar}[1]{$^{#1}$}

\addcontentsline{toc}{section}{References}

\end{document}